\documentclass[conference, twocolumn]{IEEEtran}

\IEEEoverridecommandlockouts

\usepackage[T1]{fontenc}

\ifCLASSINFOpdf
\else
\fi
\usepackage{caption}
\usepackage{subcaption}
\usepackage{cite}
\usepackage{amsmath,amssymb,amsfonts}
\usepackage{algorithmic}
\usepackage{hyperref}
\hypersetup{bookmarksopen}
\usepackage{booktabs}
\usepackage{graphicx}
\usepackage{textcomp}
\usepackage{array,tabularx, makecell, multirow}
\usepackage{xcolor}
\usepackage{diagbox,multirow,lineno}
\usepackage{amsthm}
{
	
	\newtheorem{corollary}{Corollary}

	\newtheorem{lemma}{Lemma}

}
\usepackage[cmintegrals]{newtxmath}

\begin{document}
%
\title{Age of Information in Reservation Multi-Access Networks with Stochastic Arrivals}
%
%
%

\author{\IEEEauthorblockN{Qian Wang and He Chen}
	  \thanks{The work of Q. Wang is supported in part by the Research Talent Hub PiH/456/21 under Project ITS/204/20. The work of H. Chen is supported in part by the Innovation and Technolgy Fund (ITF) under Project ITS/204/20 and the CUHK direct grant for research under Project 4055126.}
	\IEEEauthorblockA{Department of Information Engineering, The Chinese University of Hong Kong, Hong Kong SAR, China} 
	Email: \{qwang, he.chen\}@ie.cuhk.edu.hk}
\maketitle

\begin{abstract}
This paper investigates the Age of Information (AoI) performance of Frame Slotted ALOHA with Reservation and Data slots (FSA-RD). We consider a symmetric multi-access network where each user transmits its randomly generated status updates to an access point in a framed manner. Each frame consists of one reservation slot and several data slots. The reservation slot is made up of some mini-slots. In each reservation slot, users, with a status update packet to transmit, randomly send short reservation packets in one of the mini-slots to contend for data slots of the frame. The data slots are assigned to those users that succeed in reservation slot. To provide insights in optimizing the information freshness of FSA-RD, we manage to derive a closed-form expression of the average AoI under FSA-RD by applying a recursive method. Numerical results validate the analytical expression and demonstrate the influence of the frame size and reservation probability on the average AoI. We finally perform a comparison between the AoI performance of FSA-RD with optimized frame size and reservation probability, and that of slotted ALOHA with optimized transmission probability. The comparison results show that FSA-RD can effectively reduce the AoI performance of multi-access networks, especially when the status arrival rate of the network becomes large.  
\end{abstract}


%

\section{Introduction}
Recent years have witnessed increasing research interest in a new performance metric, Age of Information (AoI), thanks to its capability of quantifying the timeliness of data transmission in status update systems \cite{kaul2012real,yates2017status}. The timeliness of data transmission plays an important role in various Internet of Thing (IoT) applications, particularly in real-time monitoring systems. In these systems, the dynamics of the monitored process should be grasped at the monitor in a timely manner to ensure in-time responses. AoI is defined as the time elapsed since the generation time of the latest received status update at the receiver\cite{kaul2012real}. According to its definition, AoI is jointly determined by the transmission interval, transmission delay, and status generation process.

The trend of massive connectivity of IoT networks \cite{massiot} and the importance of information freshness have recently attracted considerable efforts in optimizing information freshness of multi-access networks. In this context, how to dynamically schedule the data transmission of users in multi-access networks to minimize the network-wide average AoI becomes a critical problem. Polling-based centralized scheduling schemes \cite{kadota2018scheduling,kadota2019minimizing,8437712,bedewy2019age} and contention-based random access schemes \cite{unma,slotaloha,gurandom,xrchen,9488897,maatouk2019minimizing1,modernra} are two major research branches. The sporadic IoT traffic makes contention-based random access schemes more preferable in large-scale networks. This is because centralized scheduling is usually associated with excessive overhead and high operation complexity. In contrast, contention-based multi-access schemes have acceptable overhead with simple operation, and can flexibly adapt to the networks with a varying number of devices.

Previous studies on contention-based multi-access schemes have explored the average AoI of ALOHA and CSMA protocols. The average AoI of slotted ALOHA systems was first characterized in \cite{unma}. Specifically, the ALOHA-alike policy, in which each user transmits its status updates with a fixed transmission probability was analyzed and compared with the centralized scheduling policy \cite{unma}. Inspired by this seminal work, age-dependent slotted ALOHA policy was devised and analyzed in \cite{slotaloha,gurandom}, where each user randomly accesses the shared channel only when its instantaneous AoI exceeds a predetermined threshold. Different from previous studies considering \textit{generate-at-will} status generation model, the authors in \cite{xrchen,9488897,maatouk2019minimizing1} investigated the contention-based multi-access schemes with stochastic arrival models for status updates. In \cite{xrchen,9488897}, status generation at each user is modeled by independent and identically distributed (i.i.d.) Bernoulli process. The analytical AoI performance of a stationary randomized policy was derived, and the asymptotic optimality of slotted ALOHA with small status arrival rate in the regime of infinite users was analyzed in \cite{xrchen}. An age-based thinning method was proposed in \cite{xrchen} to further improve the AoI performance of the slotted ALOHA systems. The authors in \cite{9488897} derived approximate expressions for the average AoI of both slotted ALOHA and CSMA schemes by developing a discrete-time model. The AoI performance of CSMA was also studied in \cite{maatouk2019minimizing1}, where the stochastic hybrid system was applied to derive the accurate average AoI expression for \textit{generate-at-will} model and a tight upper bound of average AoI for stochastic arrival model, respectively. Very recently, the AoI performance of irregular repetition slotted ALOHA was studied and compared with that of slotted ALOHA in \cite{modernra}.

{Apart from the polling-based and contention-based multi-access schemes, there is another type of dynamic allocation schemes, called reservation-based multi-access (R-MA) \cite{Bing2002}. In each frame of R-MA, one reservation slot, consisting of multiple mini-slots, is introduced for users to contend for the data slots. Only those users made successful reservations are allowed to transmit in their reserved data slots of the said frame, leading to non-conflicting data transmission \cite{roberts1973dynamic,szpankowski1983analysis,CASARESGINER201915}. Furthermore, R-MA schemes squeeze the potential collisions among users into the shorter mini-slots, reducing the time overhead for contention. In light of these features, R-MA scheme have a great potential to reduce the network-wide AoI. Nevertheless, to the best of the authors' knowledge, the AoI performance of R-MA protocols has not been thoroughly characterized in open literature.}
	

As an attempt to fill the gap, in this paper we investigate the average AoI performance of Frame Slotted ALOHA with Reservation and Data slots (FSA-RD), which is an representative R-MA protocol \cite{roberts1973dynamic,szpankowski1983analysis,CASARESGINER201915}. Specifically, we consider a symmetric multi-access network, where each user transmits its randomly generated status updates to an access point (AP) in a framed manner. The stochastic arrivals of status updates, the randomness in reservation attempts and reservation slot selection (i.e., which slot to reserve), as well as the tangled reservation results (i.e., whether the reservation is successful) altogether make the theoretical analysis of the average AoI (AAoI) for the considered system non-trivial. To provide insights in optimizing the information freshness of FSA-RD, we manage to derive a closed-form expression of the AAoI. Specifically, we focus on evaluating AoI evolution of a particular user to analyze the network-wide AAoI. A recursive method is applied to characterize all possible combinations of status generations, packet preemptions, and reservation attempts for deriving AAoI under FSA-RD. Numerical results are provided to validate the analytical expression of AAoI, and to evaluate the influence of frame size and reservation probability on AAoI. We finally perform a comparison between AAoI of FSA-RD with optimized frame size and reservation probability, and that of slotted ALOHA with optimized transmission probability. Our results show that FSA-RD can substantially reduce the AAoI of multi-access networks, especially when the status arrival rate of the network becomes large. 
\section{System Model}\label{sec2}
\subsection{System Description}
We consider a multi-access network, where $N$ users share a wireless channel to transmit time-sensitive information to an AP in a framed manner. The Frame Slotted ALOHA with Reservation and Data slots scheme (FSA-RD) \cite{roberts1973dynamic,szpankowski1983analysis} is adopted. Specifically, each frame consists of $M$ slots. At the beginning of each time slot, the information source of each user randomly generates a time-stamped status update, which is modeled by a Bernoulli process.  The network is considered to be symmetric, and the status generation probabilities for all users are equal and denoted by $\rho$. Users are frame-synchronized. 

We follow \cite{modernra} and consider
that a newly generated status update is allowed to be transmitted in the subsequent frame. In this context, more than one status updates may arrive at each user before the user being able to access the channel in the next frame. To maintain the information freshness, only the status update most recently generated during a frame will be kept in the buffer and transmitted in the next frame. As such, in any frame $k$, whether a user has a status update to transmit also follows a Bernoulli distribution. The corresponding probability equals to the probability that there is at least one status update generated in previous $M$ slot during frame $k-1$. Let $I_n(k)$ denote the indicator that equals $1$ if user $n$ has one status update to transmit in frame $k$, and equals to $0$ otherwise. We then have $\mathrm{Pr}(I_n(k)=1)=1-(1-\rho)^M$ , and the value of $I_n$ in each frame is i.i.d.  

\begin{figure}[htbp]
	\centering \scalebox{0.5}{\includegraphics{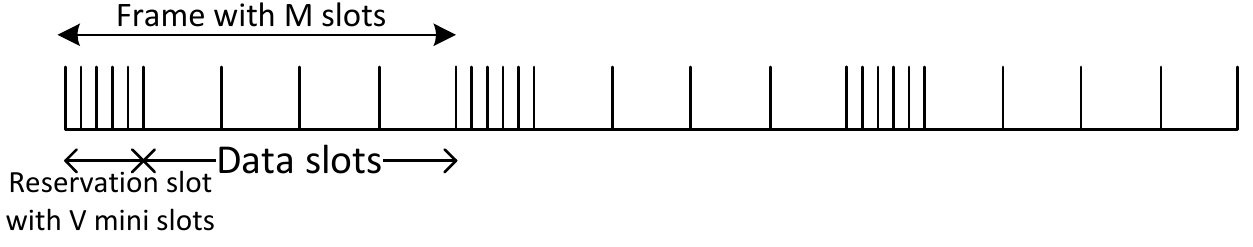}}
	\caption{Frame structure.}
	\label{fig1}
	\vspace{-1.2em}
\end{figure}
For FSA-RD, the $M$ slots in each frame can be divided into one reservation slot and $M-1$ data slots. Specifically, the first slot of each frame is the reservation slot, used for making reservations, while the rest $M-1$ data slots are used for sending status updates, as depicted in Fig. \ref{fig1}. Each reservation slot consists of $V$ mini-slots. At the beginning of each frame, each user with a status update to transmit, will make a reservation with probability $\gamma$ (i.e., ALOHA-alike). Denote by $J_n(k)$ the reservation indicator of user $n$, which equals to $1$ if user $n$ chooses to make a reservation in frame $k$, and equals to $0$ otherwise. That is, $\mathrm{Pr}(J_n(k)=1|I_n(k)=1)=\gamma$. Once user $n$ decides to reserve, it will uniformly choose one of the $V$ mini-slots to send its reservation packet. The transmission of the reservation packet and that of the status update packet are assumed to take $1$ mini-slot and $1$ time slot, respectively, considering that the reservation packet generally contains less information than the status update packet. {If more than one user transmits reservation packets in the same mini-slot, a collision occurs and all reservations made to that mini-slot fail; otherwise, the reservation information will be gathered by the AP. At the end of the reservation slot, the AP will inform all users of reservations results\footnote{For simplicity, we ignore the duration of the feedback from AP.}}. Upon the reservation results, the first $M-1$ successfully reserved users will take turns to transmit their status updates in the data slots in the order of successful reservations. Similarly, the status update packets will be successfully received by the AP as there is no collision. Note that at most $M-1$ users can successfully update their statuses in each frame as there are totally $M-1$ data slots. If less than $M-1$ users make a successful reservation, the unreserved data slot(s) will be wasted in the current frame. In this sense, the frame size $M$ should be neither too large nor too small\footnote{The investigation for the case with variable frame length has been left as a future work.}.

\subsection{The Evolution of AoI}
The AoI of user $n$, denoted by $\delta_n(t)$, measures the timeliness of the status updates from the perspective of the AP, which is defined as the time elapsed since the generation time of the most recently received status update from user $n$ at the AP. Mathematically, the AoI $\delta_n(t)$ at time $t$ is $t-u_n(t)$, where $u_n(t)$ denotes the generation time of the latest received status update of user $n$ at the AP until time $t$. 
\begin{figure}[htbp]
	\centering \scalebox{0.4}{\includegraphics{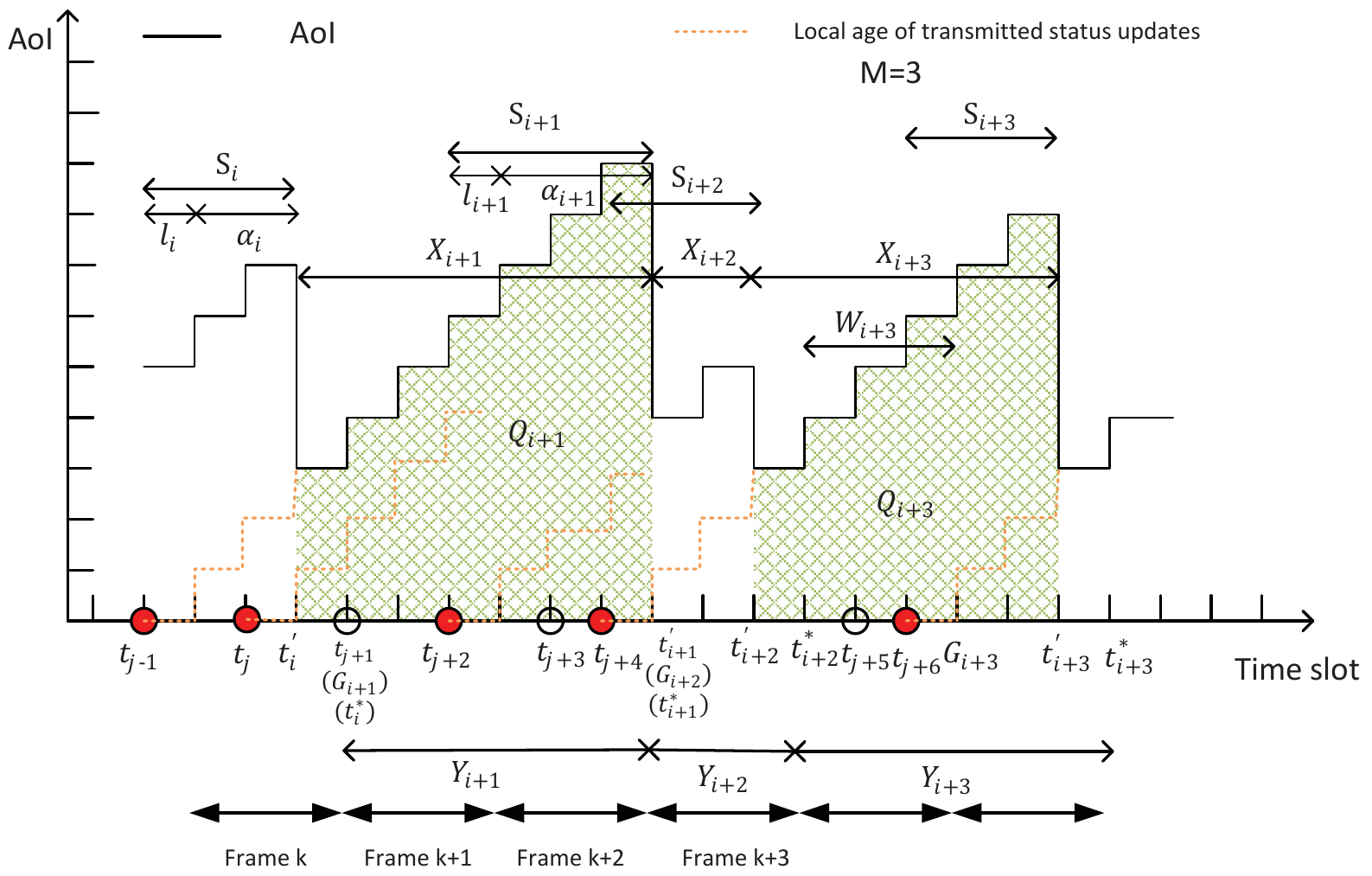}}
	\caption{The evolution of AoI with $M=3$. Each circle indicates the generation of a status update in the corresponding time slot, while the circles in red are those status updates that are transmitted. The black solid line is the AoI curve and the orange dotted line is the local age of transmitted status updates.}
	\label{fig2}
	\vspace{-1.2em}
\end{figure}

The black solid line in Fig. \ref{fig2} illustrates one example of how the AoI of user $n$ at the AP evolves in the considered system. As it can see, the AoI linearly increases until the AP receives a status update from user $n$, when it is reset to the service time of the received status update. The service time is the difference between the reception time of a new status at the AP and its generation time at user $n$. Denote by $t_j$ the generation time of $j$th status update and by $t_i^{'}$ the time of the reception of $i$th received status. The different indexes in $t_j$ and $t_i^{'}$ is because of the preemption and discard in packet management caused by reservation failures and no-reservation attempts. As shown in Fig. \ref{fig2}, only those status updates that arrived at $t_{j-1}$, $t_{j+2}$, $t_{j+4}$ and $t_{j+6}$, were successfully received by the AP at $t_i^{'}$, $t_{i+1}^{'}$, $t_{i+2}^{'}$ and $t_{i+3}^{'}$, respectively. The status update that arrived at $t_j$ was transmitted in frame $k+1$, but it might suffer from either reservation failure or no-reservation attempt, and thus was discarded. The remaining status updates were replaced by late arrived status updates. According to the definition of $u(t)$, the service time of $i$th received status update can be expressed as $S_i=t_i^{'}-u(t_i^{'})$. Recall that each status update can be transmitted only in the next frame. As such, status update packets may need to wait in the buffer before being transmitted. We further introduce the need to define what is local age of transmitted status updates, denoted by $g_n(t)$. If the AP receives a status update transmitted from user $n$ in time slot $t$, then $\delta_n(t+1)=g_n(t)+1$; otherwise, $\delta_n(t+1)=\delta_n(t)+1$. Other variables showed in Fig. \ref{fig2} will be explained in the following section for further evaluation of AoI.%
\section{Analysis of AAoI}
 We note that the AoI of each user is identically distributed due to the symmetric network setup. As such, in this section we focus on analyzing the AAoI of one particular user to represent the network-wide AAoI performance and omit user index for brevity. To derive the analytical expression of AAoI, we first elaborate some useful definitions for the considered system, inspired by the analytical framework presented in \cite{gu2019timely}. 
 
 {We define $t_{i}^{*}$ as the end of the frame within which the $i$th status update is successfully received at the AP, which is also the beginning of the next frame after $i$th reception of status updates. As the status updates randomly arrive at each user and are transmitted in the next frame, we define the waiting time $W_i$ as the time elapsed from $t_{i-1}^{*}$, until the beginning of the first frame when the user has a status update to transmit after $t_{i-1}^{*}$, denoted by $G_{i}$. 
We thus have $W_i=G_i-t_{i-1}^{*}$. Due to the frame-based  transmission structure, $W_i$ takes value from $\{0,M,2M,...\}$. Recall that whether a user has a status update for transmission in each frame  follows a Bernoulli distribution with parameter $p=1-(1-\rho)^M$. We then can obtain the probability mass function (PMF) of $W_i$, given by $\mathrm{Pr}\{W_i=xM\}=(1-p)^xp$, and we further have 
\begin{equation}\label{eqw}
\mathbb{E}[W_i]=(1-p)M/p,
\end{equation}
\begin{equation}\label{eqw2}
\mathbb{E}[W_i^2]=(p^2-3p+2)M^2/p^2.
\end{equation} We also define $K_i$ as the time interval from $G_i$ until the end of the frame within which the $i$th status update is received at the AP, i.e., $K_i=t_{i}^{*}-G_i$. We further define $Y_i$ as the time interval between the ends of two frames within which the $(i-1)$th and $i$th status updates are received at the AP, respectively, i.e., $Y_i=t_{i}^{*}-t_{i-1}^{*}$. Together with the definitions of $W_i$ and $K_i$, we have $Y_i=W_i+K_i$.}

Let $X_i$ denote the inter-departure time of two successive correctly received status updates at the AP, i.e., $X_i=t_{i}^{'}-t_{i-1}^{'}$.
Let $N_t$ denote the number of status updates that have been successfully received by the AP until time $t$. Similar to \cite[Eq 2]{kaul2012real}, AAoI can be expressed as 
\begin{equation}
\label{eqa11}
\bar{\Delta}=\lim_{t\rightarrow \infty}\frac{N_t}{t}\frac{1}{N_t}\sum_{i=1}^{N_t}Q_i=\frac{\mathbb{E}[Q_i]}{\mathbb{E}[X_i]},
\end{equation} where $Q_i$ is the polygon area as depicted in Fig. \ref{fig2}. The area of $Q_i$ can be calculated as
\begin{equation}
\label{eqa12}
Q_i=S_{i-1}+S_{i-1}+1+...+S_{i-1}+X_i-1={\left(X_{i}^2-X_{i}\right)}/{2}+S_{i-1}X_i.
\end{equation} By taking the expectations on both sides of \eqref{eqa12} {and substituting the expectations into \eqref{eqa11}, we have
\begin{equation}
\label{eqa13*}
\bar{\Delta}=\frac{\mathbb{E}[X_i^2]}{2 \mathbb{E}[ X_i]}+\frac{\mathbb{E}[S_{i-1}X_i]}{\mathbb{E}[ X_i]}-\frac{1}{2}.
\end{equation} 

According to the status generation model and transmission model, $S_{i-1}$ consists of two time intervals that a successfully received status update has experienced in its non-preemptive generation frame and transmission frame until its reception, denoted by $l_{i-1}$ and $\alpha_{i-1}$, respectively. Mathematically, $S_{i-1}=l_{i-1}+\alpha_{i-1}$, where $l_{i-1}$ describes the status generation process while $\alpha_{i-1}$ characterizes the status transmission process. As such, $l_{i-1}$ and $\alpha_{i-1}$ are independent. Meanwhile, the inter-departure time ($X_{i}$) between the $(i-1)$th and $i$th receptions of status updates consists of three parts: the remaining time slots in the transmission frame of the $(i-1)$th received status update since its reception (i.e., $M-\alpha_{i-1}$); the frames without successful status transmission since the $(i-1)$th reception of status updates (i.e., $Y_{i}-M$); and the time slots in the transmission frame of the $i$th received status update until its reception (i.e., $\alpha_{i}$). That is, $X_{i}=M-\alpha_{i-1}+Y_{i}-M+\alpha_{i}$. Recall that $\alpha_{i-1}$ and $\alpha_{i}$ are the time intervals that the $(i-1)$th and $i$th received status updates spent in their transmission frame until reception, and thus they are i.i.d. Thus, we have $\mathbb{E}[X_i]=\mathbb{E}[Y_i]$.

Besides, each status update can be transmitted at most once in the next frame after its generation and all transmissions of one user are independent. As such, the unsuccessful frame interval $Y_{i}-M$ is also independent of $\alpha_{i-1}$ and $\alpha_{i}$ as well as $l_{i-1}$. Together with the i.i.d. property of the sequence of random variable $\{\alpha_i\}$, we have
\begin{equation}\label{ex2}
\mathbb{E}[X_i^2]=\mathbb{E}[Y_i^2]+2\mathrm{Var}(\alpha),
\end{equation}
\begin{equation}\label{esx}
\mathbb{E}[S_{i-1}X_i]=\mathbb{E}[S_{i-1}]\mathbb{E}[Y_i]-\mathrm{Var}(\alpha),
\end{equation} where $\mathrm{Var}(\alpha)$ is the variance of random variable $\alpha_i$ ($\alpha_{i-1}$ also). By substituting \eqref{ex2} and \eqref{esx} into \eqref{eqa13*}, together with $\mathbb{E}[X_i]=\mathbb{E}[Y_i]$, we have 
\begin{equation}
\label{eqa13}
\bar{\Delta}=\frac{\mathbb{E}[Y_i^2]}{2 \mathbb{E}[ Y_i]}+\mathbb{E}[S_{i-1}]-\frac{1}{2}.
\end{equation} Due to the i.i.d. property of the two sequences of random variables $\{S_i\}$ and $\{Y_i\}$, we hereafter omit the subscripts of $S_{i-1}$ and $Y_i$, and calculate $\mathbb{E}[S]$, $\mathbb{E}[Y]$ and $\mathbb{E}[Y^2]$ for brevity.}
\subsection{ Evaluation of $\mathbb{E}[S]$}
We first calculate the expected service time $\mathbb{E}[S]$. Service time only counts those successfully transmitted status updates. To proceed, we analyze the successful status update probability for user $n$, denoted by $p_s$, once it decides to make a reservation in a frame when a status update is available to transmit. Let $I_a^n(k)=1$ denote the successful reception of a status update from user $n$ at the AP within frame $k$ . Then we have $p_s=\mathrm{Pr}\left(I_a^n=1|J_n=1, I_n=1\right)$, where we omit the frame indexes of $J_n$, $I_n$ and $I_a^n$, as they are frame-independent. The following lemma gives a closed-form expression of $p_s$,
\begin{lemma}\label{lem1}
	$p_s=\sum_{n_1=0}^{N-1}\sum_{n_2=0}^{n_1}\sum_{n_3=1}^{\min\{V,n_2+1\}}C_{N-1}^{n_1}(1-p)^{N-1-n_1}p^{n_1}C_{n_1}^{n_2}(1-\gamma)^{n_1-n_2}\gamma^{n_2}\frac{\min\{n_3,M-1\}}{n_2+1}\frac{(-1)^{n_3}V!(n_2+1)!}{V^{(n_2+1)}n_3!}\times\\
	\left\{\sum_{m=n_3}^{\min\{V,(n_2+1)\}}\frac{(-1)^m(V-m)^{(n_2+1)-m}}{(m-n_3)!(V-m)!(n_2+1-m)!}\right\}
	$.
\end{lemma}
\begin{proof}
	Note that $p_s$ depends on the number of users excluding user $n$ that decide to make reservations in a certain frame, denoted by $N_r$. We have $N_r\in\{0,1,...,N_g\}$, where $N_g$ denotes the number of users excluding user $n$ that have a status update packet to transmit, and thus $N_g\in\{0,1,...,N-1\}$. According to the status generation model and reservation scheme, the PMF of $N_g$ and the conditional PMF of $N_r$ given $N_g=n_1$ can be expressed as  $\mathrm{Pr}\left\{N_g=n_1\right\}=C_{N-1}^{n_1}(1-p)^{N-1-n_1}p^{n_1}$ and $\mathrm{Pr}\left(N_r=n_2|N_g=n_1\right)=C_{n_1}^{n_2}(1-\gamma)^{n_1-n_2}\gamma^{n_2}$, respectively. Together with user $n$, there will be $N_r+1$ users randomly selecting one of the $V$ min-slots to send a short reservation packet. The number of min-slots that are reserved by one single user, denoted by $N_s$, depends on the total number of reservation users.  According to \cite[Eq. 6]{szpankowski1983analysis}, we have $\mathrm{Pr}\left(N_s=n_3|N_r=n_2\right)=\frac{(-1)^{n_3}V!(n_2+1)!}{V^{(n_2+1)}n_3!}\sum_{m=n_3}^{\min\{V,(n_2+1)\}}\frac{(-1)^m(V-m)^{(n_2+1)-m}}{(m-n_3)!(V-m)!(n_2+1-m)!}$ denoting the probability that $n_2+1$ users successfully reserve $n_3$ mini-slots. Due to the identical reservation scheme of each user and the limitation of $M-1$ data slots, the successful reservation probabilities of the $n_2+1$ users are the same. For any of the $n_2+1$ users, the probability of successful status update is $\frac{\min\{n_3,M-1\}}{n_2+1}$, i.e., $\mathrm{Pr}\left\{I_a^n=1|N_s=n_3,N_r=n_2\right\}=\frac{\min\{n_3,M-1\}}{n_2+1}$. Based on all the above analysis and the law of total probability, we arrive at the expression of $p_s$ given in Lemma 1. This completes the proof.
\end{proof}
{{The following corollary characterizes the probability that one user successfully transmits its status update in the reserved $(\alpha-1)$th data slot, denoted by $\varphi_\alpha$, where $\alpha\in\{2,...,M\}$ and $\sum_{\alpha=2}^M\varphi_{\alpha}=p_s$. The proof is omitted as it can be directly inferred from Lemma \ref{lem1}.
\begin{corollary}\label{co1}
		$\varphi_\alpha=\sum_{n_1=0}^{N-1}\sum_{n_2=0}^{n_1}\sum_{n_3=\alpha-1}^{\min\{V,n_2+1\}}C_{N-1}^{n_1}(1-p)^{N-1-n_1}p^{n_1}C_{n_1}^{n_2}(1-\gamma)^{n_1-n_2}\gamma^{n_2}\frac{1}{n_2+1}\frac{(-1)^{n_3}V!(n_2+1)!}{V^{(n_2+1)}n_3!}\times\\
	\left\{\sum_{m=n_3}^{\min\{V,(n_2+1)\}}\frac{(-1)^m(V-m)^{(n_2+1)-m}}{(m-n_3)!(V-m)!(n_2+1-m)!}\right\}
	$.
\end{corollary}

Recall that $S=l+\alpha$, and thus $\mathbb{E}[S]=\mathbb{E}[l]+\mathbb{E}[\alpha]$. 
Due to the independent generation of status updates in each slot, we have $l\in\{1,2,...,M\}$ and the probability that a status update is generated $l$ slots before its transmission frame without preemption is given by $\phi_{l}=\rho(1-\rho)^{l-1}$, which is the product between the probability of one status generation at the $(M-l)$th slot of a frame and the probability of no status generation in the following $l-1$ consecutive slots. We then have
\begin{equation}
\mathbb{E}[l]=\frac{\sum_{l=1}^{M}\phi_{l}l}{\sum_{l=1}^{M}\phi_{l}}=\frac{1}{\rho}-\frac{M(1-\rho)^M}{1-(1-\rho)^M}.
\end{equation}
As for $\mathbb{E}[\alpha]$, based on Corollary \ref{co1}, we have 
$\mathbb{E}[\alpha]={\sum_{\alpha=2}^{M}\varphi_{\alpha}\alpha}/{p_s}$. In this regard, $\mathbb{E}[S]$ can be expressed as 
\begin{equation}\label{es}
\mathbb{E}[S]=\frac{1}{\rho}-\frac{M(1-\rho)^M}{1-(1-\rho)^M}+\frac{\sum_{\alpha=2}^{M}\varphi_{\alpha}\alpha}{p_s}.
\end{equation}}
\subsection{Evaluation of $\mathbb{E}[Y]$ and $\mathbb{E}[Y^2]$}
We note that $W$ and $K$ are independent because $W$ only depends on status arrival rate $\rho$. Recall that $Y=W+K$, we thus have}
\begin{equation}
\label{eq10}
\mathbb{E}[ Y]=\mathbb{E}[W]+\mathbb{E}[K],
\end{equation}
\begin{equation}
\label{eq11}
\mathbb{E}[Y^2]=\mathbb{E}[(W+K)^2]=\mathbb{E}[W^2]+\mathbb{E}[K^2]+2\mathbb{E}[W]\mathbb{E}[K],
\end{equation} where $\mathbb{E}[W]$ and $\mathbb{E}[W^2]$ are given in \eqref{eqw} and \eqref{eqw2}.

It is not easy to directly calculate the distribution of $K$. Inspired by \cite{gu2019timely,chen2016age}, we apply a recursive method to calculate the expectations of $K$ and $K^2$. We note that the term $K$ has two different behaviors depending on whether the status update is successfully received by the AP: \textbf{1)} If the status update is successfully received in its transmission frame, $K=M$ with probability $\gamma p_s$; \textbf{2)} If not, the user needs to wait for the generation of a new status update to transmit. Then, $K=M+W^{'}+K^{'}$ with probability $(1-\gamma p_s)$. Here, $W^{'}$ denotes the waiting time of a new status available to transmit and $K^{'}$ is the remaining frames to successfully transmit such a new status update. We notice that $\mathbb{E}[K]=\mathbb{E}[K^{'}]$ due to the same evolution, and $\mathbb{E}[W]=\mathbb{E}[W^{'}]$ due to the i.i.d. process. Then, $\mathbb{E}[K]$ can be calculated as following,
\begin{equation}
\label{ek}
\begin{split}
\mathbb{E}[K]=\gamma p_sM+(1-\gamma p_s)(M+\mathbb{E}[W]+\mathbb{E}[K]).
\end{split}
\end{equation} After some manipulations, we have 
\begin{equation}
\label{ek2}
\mathbb{E}[K]=\frac{M+(1-\gamma p_s)\mathbb{E}[W]}{\gamma p_s}.
\end{equation} Then, the expectation of inter-departure time of two successive correctly transmitted status updates is 
\begin{equation}
\label{ey}
\mathbb{E}[Y]=\mathbb{E}[W]+\mathbb{E}[K]=\frac{M}{\gamma p_s}+\frac{M(1-\rho)^M}{\gamma p_s\left(1-(1-\rho)^M\right)}.
\end{equation}

As for the expected value of $Y^2$, according to \eqref{eq11}, we need to calculation $\mathbb{E}[K^2]$. Using the same method as calculating $\mathbb{E}[K]$ in \eqref{ek}, $\mathbb{E}[K^2]$ can be calculated as
\begin{equation}\label{mek2}
	\begin{split}
	&\mathbb{E}[K^2]=\gamma p_sM^2+(1-\gamma p_s)\left(\!M^2\!+\!\mathbb{E}[K^2]\!+\!\mathbb{E}[W^2]\right)+
	\\
	&(1-\gamma p_s)\left(2\mathbb{E}[K]\mathbb{E}[W]+2M\left(\mathbb{E}[K]+\mathbb{E}[W]\right)\right).
	\end{split}
\end{equation} By substituting $\mathbb{E}[W^2]$ and $\mathbb{E}[K]$ into \eqref{mek2}, $\mathbb{E}[K^2]$ can be obtained. According to \eqref{eq11}, we have $\mathbb{E}[Y^2]$ as follows,
\begin{equation}
\!\mathbb{E}[Y^2]\!\!=\!\!\frac{M^2}{\gamma p_s}+\frac{2M^2\!\!\left(\!\frac{1}{\gamma p_s}-\left(1-(1-\rho)^M\right)\!\right)}{\gamma p_s\left(1-(1-\rho)^M\right)^2}+\frac{M^2(1-\rho)^M}{\gamma p_s(1-(1-\rho)^M)}
\end{equation}
Based on $\mathbb{E}[Y^2]$, $\mathbb{E}[Y]$ and $\mathbb{E}[S]$, we obtain the expression of AAoI of FSA-RD, given by
\begin{equation}\label{age}
\bar{\Delta}\!\!=\!\!\frac{M}{\gamma p_s\left(1-(1-\rho)^M\right)}-\frac{M(1-\rho)^M}{1-(1-\rho)^M}+\frac{1}{\rho}-\frac{M+1}{2}+\frac{\sum_{\alpha=2}^{M}\varphi_{\alpha}\alpha}{p_s},
\end{equation} where $p_s$ and $\varphi_{\alpha}$ are given in Lemma \ref{lem1} and Corollary \ref{co1}, respectively.
\section{Numerical results}\label{nr}
In this section, both numerical and analytical results of AoI performance of the considered system are presented. We first evaluate the derived analytical expression of AAoI by comparing it with Monte Carlo simulation results. Fig. \ref{fig3} plots the AAoI curves versus reservation probability $\gamma$, considering a multi-access system with $N=30$ and $V=4$. Firstly, we can see from Fig. \ref{fig3} that our analytical results coincide well with the corresponding simulation results, which validates our derivation of $\bar{\Delta}$. Secondly, we can find that when the status generation probability goes large, the optimal reservation probability should be neither too large nor too small for better AoI performance. This is understandable as a larger reservation probability is more likely to cause reservation failures (i.e., more reservation collisions), while a smaller reservation probability makes users less likely to make a reservation. In both cases, fewer users transmit status updates in data slots and some unreserved data slots are wasted, leading to larger network-wide AAoI performance. Besides, when the status generation probability is considerably small, the reservation probability should be as large as possible. The rationale is that when status updates are rarely generated, it is less likely to cause collision even when all users make reservations once they have a status update to transmit. In contrast, small reservation probability may lead to the drop of rarely generated status updates, resulting in larger AoI. Moreover, by comparing the performance of different frame length $M$, we observe that larger $M$ in some cases may cause performance degradation in FSA-RD. This is because when only small number of users make successful reservations, after these users transmitting their status updates, the rest data slots in the frame will be wasted. In this sense, the choice of $M$ should take the successful reservation rate into consideration.
\begin{figure}[!t]
	\centering \scalebox{0.35}{\includegraphics{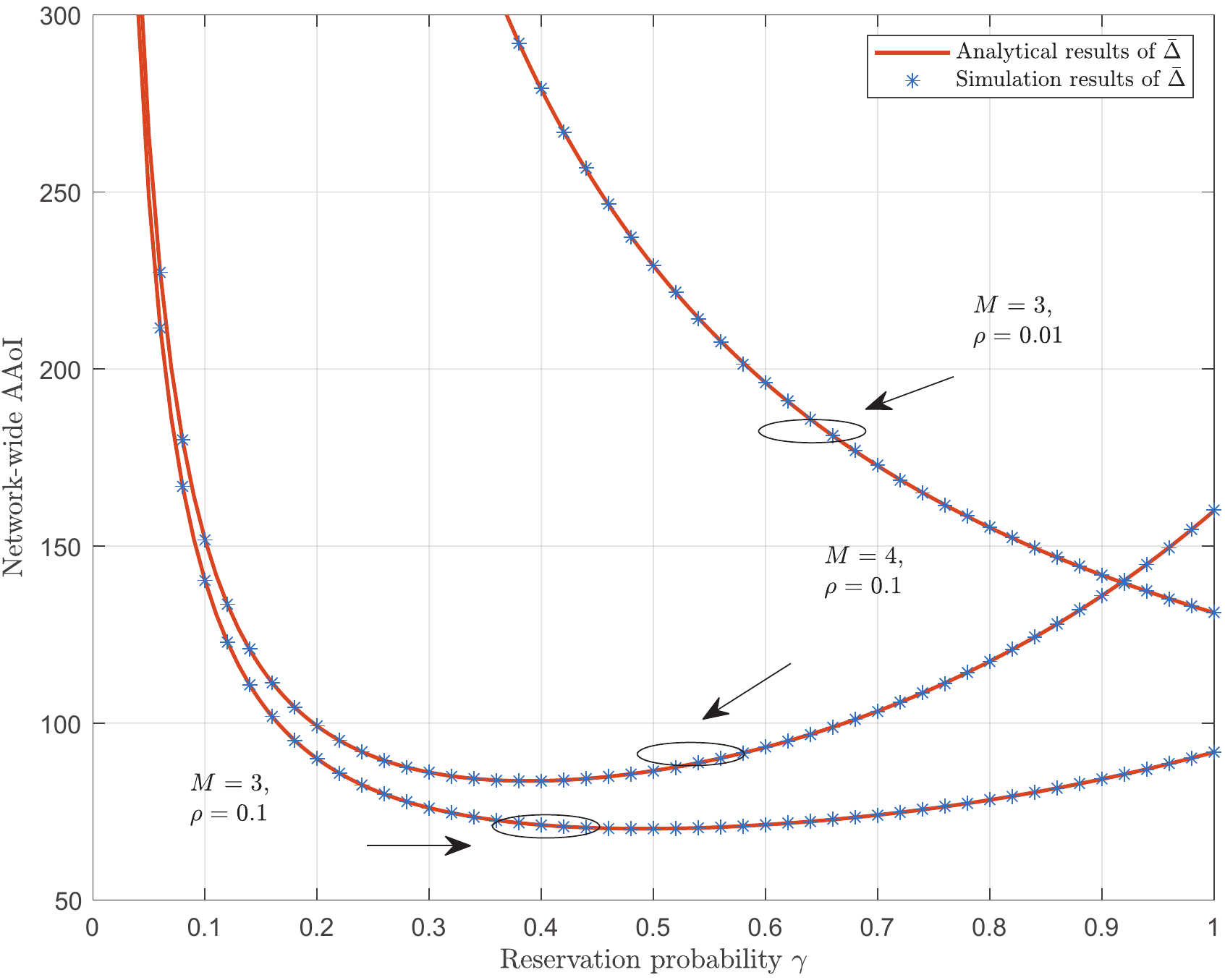}}
	\caption{Network-wide AAoI of a $30$-user system with $V=4$ versus reservation probability $\gamma$ for different frame length $M$ and status generation probability $\rho$.}
	\label{fig3}
	\vspace{-1.5em}
\end{figure}

Table \ref{table1} compares the of FSA-RD under the optimal frame length $M$ and optimal reservation probability $\gamma$, with the simulation results of the optimized AAoI performance of slotted ALOHA. Specifically, given the network setups, we exhaustively search all values of $M$ and $\gamma$ to achieve the minimum value of $\eqref{age}$ for RSA-RD. As for slotted ALOHA, we simulate the age evolution with different values of transmission probability to find its optimal AAoI performance. We can see that when the product $\rho N$ is small, the slotted ALOHA has smaller AAoI than FSA-RD. As $\rho N$ increases through either increasing $N$ or $\rho$, the optimal performance of FSA-RD improves and becomes better than that of slotted ALOHA with up to $29\%$ AAoI reduction. The observation could be caused by the joint effect of the one slot reservation overhead and the frame-based transmission (i.e., status update will be transmitted at the next frame and dropped after its transmission frame). Besides, it is obvious that a larger number of mini-slots (i.e., the value of $V$) in reservation slot leads better AAoI performance in FSA-RD scheme as the larger the value of $V$, the higher the successful reservation probability.  
\begin{table}[h]
	\begin{subtable}[h]{0.4\textwidth}
		\centering
		\begin{tabular}{|c|*{4}{c|} }
		\hline
		& $\rho=0.01$&$\rho=0.02$&$\rho=0.04$&$\rho=0.08$\\
		\hline
		FSA-RD, $V=4$ & $131.16$& $86.46$&\textbf{70.74}&\textbf{70.18}\\
		\hline
		FSA-RD, $V=6$ & $124.06$& \textbf{78.74}&\textbf{60.42}&\textbf{56.47}\\
		\hline
		slotted ALOHA& $110.14$&$82.55$&$81.30$&$80.22$\\
		\hline	
		\end{tabular}	
		\caption{Network-wide AAoI versus packet transmission probability $\rho$ with $N=30$.}
		\label{tab:week1}
	\end{subtable}
	\hfill
	\begin{subtable}[h]{0.4\textwidth}
		\centering
		\begin{tabular}{|c|*{4}{c|} }
		\hline
		& $N=10$&$N=20$&$N=40$&$N=50$\\
		\hline
		FSA-RD, $V=4$ & $37.40$& \textbf{52.12}&\textbf{93.12}&\textbf{116.04}\\
		\hline
		FSA-RD, $V=6$ & $35.12$& \textbf{46.63}&\textbf{75.89}&\textbf{92.90}\\
		\hline
		slotted ALOHA& $31.63$&$53.72$&$107.66$&$136.97$\\
		\hline	
		\end{tabular}	
		\caption{Network-wide AAoI versus total number of users $N$ with $\rho=0.04$.}
		\label{tab:week2}
	\end{subtable}
	\caption{Performance comparison between optimized slotted ALOHA and FSA-RD.}
	\label{table1}
	\vspace{-1.8em}
\end{table}
\section{Conclusions}
We investigated the average age of information (AAoI) performance of Frame Slotted ALOHA with Reservation and Data slots (FSA-RD). A symmetric multi-access network was considered, where each user transmits its randomly generated status updates to an access point in a framed manner. To gain insights in optimizing the AAoI of FSA-RD, we derived a closed-form expression of the AAoI under FSA-RD. The correctness of the analytical AAoI was verified by comparing with the simulated AAoI in numerical results. Numerical results showed that when the status arrival rate of the network becomes large, FSA-RD with optimized reservation probability and frame size can achieve up to $29\%$ AAoI reduction, compared with optimized slotted ALOHA. Future work will include the investigation of the AAoI of FSA-RD with variable frame length and different contention schemes for making reservations. 


%

\ifCLASSOPTIONcaptionsoff
  \newpage
\fi



%
\bibliography{ref}
\bibliographystyle{IEEEtran}

%




\end{document}